\documentclass[11pt]{article}
\usepackage[margin=1in]{geometry}
\usepackage{mdframed}
\usepackage{amsmath}
\usepackage{amsfonts}
\usepackage{amssymb}
\usepackage{amsthm}
\usepackage{newlfont}
\usepackage{mathtools,array}
\usepackage[dvipsnames]{xcolor}
\usepackage{graphicx}
\usepackage{mathtools}
\usepackage{amsthm}
\usepackage{hyperref}
\usepackage{cleveref}

\newcommand{\ceil}[1]{\lceil #1 \rceil}

\newcommand{\etal}{et al.\ }

\title{Algorithms for Intersection Graphs of Multiple Intervals and Pseudo-Disks}

\author{Chandra Chekuri \footnote{Chandra Chekuri is supported in part by NSF grants
		CCF-1526799 and CCF-1910149.} \\ University of Illinois, Urbana-Champaign \\ \texttt{chekuri@illinois.edu} \and Tanmay Inamdar \footnote{Tanmay Inamdar is supported in part by NSF grant CCF-1615845.} \\ The University of Iowa \\ \texttt{tanmay-inamdar@uiowa.edu}}
\date{}

\newcommand{\B}{\mathcal{B}}
\newcommand{\R}{\mathcal{R}}
\newcommand{\I}{\mathcal{I}}
\ifdefined\S
\renewcommand{\S}{\mathcal{S}}
\else
\newcommand{\S}{\mathcal{S}}
\fi
\newcommand{\NP}{\mathsf{NP}}

\newtheorem{theorem}{Theorem}

\theoremstyle{definition}
\newtheorem{definition}{Definition}
\newtheorem{observation}{Observation}

\newtheorem{lemma}{Lemma}
\newtheorem{remark}{Remark}

\begin{document}

\maketitle

\begin{abstract}
	 Intersection graphs of planar geometric objects such as intervals,
	disks, rectangles and pseudo-disks are well studied. Motivated by
	various applications, Butman et al.\ \cite{ButmanHLR10} in SODA 2007 considered algorithmic
	questions in intersection graphs of $t$-intervals. A $t$-interval
	is a union of at most $t$ distinct intervals (here $t$ is a parameter)
	--- these graphs are referred to as Multiple-Interval Graphs.
	Subsequent work by Kammer et al.\ \cite{KammerTV10} also considered $t$-disks and
	other geometric shapes. In this paper we revisit some of these
	algorithmic questions via more recent developments in computational
	geometry.  For the minimum weight dominating set problem, we give a
	simple $O(t \log t)$ approximation for Multiple-Interval Graphs,
	improving on the previously known bound of $t^2$ . We also show that
	it is $\mathsf{NP}$-hard to obtain an $o(t)$-approximation in this case. In fact, our results hold for the intersection graph of a set of
	$t$-pseudo-disks which is a much larger class.  We obtain an
	$\Omega(1/t)$-approximation for the maximum weight independent set in
	the intersection graph of $t$-pseudo-disks. Our results are based on
	simple reductions to existing algorithms by appropriately bounding the
	union complexity of the objects under consideration.
\end{abstract}

\section{Introduction}
A number of interesting optimization problems can be modeled as
packing and covering problems involving geometric objects in the plane
such as intervals, disks, rectangles, triangles, convex objects and
pseudo-disks. The intersection graphs of geometric objects are also of
much interest for both theoretical and practical reasons.  For
instance the famous Koebe-Andreev-Thurston theorem shows that every
planar graph can be represented as the intersection graph of
interior-disjoint disks in the plane \cite{thurston1979geometry}.
Interval graphs are another well-studied class of geometric
intersection graph. Such a graph is induced by a finite collection of
intervals on the real line --- each vertex of the graph represents a
(closed) interval on the real line, and there is an edge between two
vertices iff the corresponding intervals have a non-empty
intersection. Several algorithmic problems on interval graphs are
motivated by practical applications such as scheduling and resource
allocation.

Several papers in the past
\cite{Bar-YehudaHNSS06,Bar-YehudaR06,ButmanHLR10,KammerTV10} studied
generalization of geometric intersection graphs to the setting where
each (meta) object is now the union of a collection of base geometric
objects. To make the discussion concrete we first discuss $t$-interval
graphs. For an integer parameter $t$, a $t$-interval is the union of
at most $t$ intervals. A $t$-interval graph is the intersection graph
of a collection of $t$-intervals. These graphs are also called
multiple-interval graphs. They have been very well-studied from both a
graph theoretic and algorithmic point of view. For instance, every
graph with maximum degree $\Delta$ can be represented as a
$t$-interval graph for $t = \ceil{(\Delta + 1)/2}$
\cite{GriggsW80}. This demonstrates the modeling power obtained by
considering unions of simple geometric objects.  Butman \etal
\cite{ButmanHLR10}, building on \cite{Bar-YehudaHNSS06,Bar-YehudaR06}
(which primarily studied the maximum independent set problem),
considered several optimization problems over multiple-interval graphs
such as minimum vertex cover, minimum dominating set and maximum
clique.  Unlike the case of interval graphs where these problems are
tractable, the corresponding problems in multiple-interval graphs are
$\mathsf{NP}$-hard even for small values of $t$. Butmal \etal describe
approximation algorithms for these problems and the approximation
ratios depend on $t$. We refer the reader to \cite{ButmanHLR10} and
references therein for a detailed discussion of the literature and
applications of multiple-interval graphs. In a subsequent work Kammer
\etal \cite{KammerTV10,KammerT14} studied (among other results)
intersection graphs of $t$-disks (a $t$-disk is a union of at most $t$
disks) and $t$-fat objects and obtained approximation algorithms for
canonical optimization problems on these graphs such as independent
set, vertex cover and dominating set (see also \cite{YeB12}).

\medskip \noindent {\bf Motivation and our contribution:} The
contribution of this paper is to demonstrate that powerful techniques
from computational geometry that have been successfully used to
develop algorithms for packing and covering problems
\cite{ChanGKS12,ChanH12,Varadarajan10} can be easily extended to
provide algorithmic results for $t$-interval graphs and $t$-disks and
other geometric objects in a unified fashion. For some problems we
obtain substantially improved approximation bounds that are
near-optimal.  Our results extend to $t$-pseudo-disks while techniques
in earlier work that exploited properties of intervals
\cite{ButmanHLR10} or fatness properties of the underlying objects
\cite{KammerTV10} do not apply to pseudo-disks\footnote{A formal
	definition is given later in the paper. Informally a collection of
	connected regions in the plane is called a collection of pseudo-disks
	if the boundaries of any two of the regions intersect in at most two
	points. Note that the property is defined by the entire
	collection.}.

Before stating our results in full generality we consider the
following geometric covering problem. Given a collection of points on
the line and a collection of weighted intervals, find the minimum
weight subset of the given intervals that cover all the points.  This
is a special case of the Set Cover problem and can be solved
efficiently via dynamic programming or via mathematical programming
--- the natural LP in the interval case is known to be an integer
polytope since the incidence matrix between intervals and points is
totally unimodular (TUM) (it has the consecutive ones property). Now
consider the same problem where we need to cover a collection of
points by (weighted) $t$-intervals. Approximation algorithms for this
problem were considered by Hochbaum and Levin \cite{HochbaumL06} (in
the more general setting of multicover) --- they derived a relatively
straightforward $t$-approximation for this problem by reducing it,
via the natural LP relaxation, to the case of $t=1$. As far as we are
aware this was the best known approximation to this problem. A natural
question is whether the bound of $t$ can be improved. In this paper we
show that an $O(\log t)$-approximation can be obtained via tools from
computational geometry such as shallow-cell complexity and
quasi-uniform sampling.  It is an easy observation that an arbitrary
Set Cover instance where each set has at most $t$ elements can be
reduced to covering points by $t$-intervals; thus, for large values of
$t$ one obtains an $\Omega(\log t)$ hardness under $\mathsf{P} \neq \mathsf{NP}$. The
geometric machinery allows us to derive an $O(\log t)$-approximation
in the much more general setting of covering points by
$t$-pseudo-disks.

We state our results for $t$-pseudo-disks which capture several shapes
of interest including intervals and disks. The geometric approach
applies in more generality but here we confine our attention to
pseudo-disks.

\begin{itemize}
	\item An $O(\log t)$ approximation for minimum weight cover of points
	by $t$-pseudo-disks.
	\item An $O(t \log t)$ approximation for the minimum weight dominating
	set\footnote{A dominating set in a graph $G=(V,E)$ is a subset
		$S \subseteq V$ such that each node $u \in V$ is in $S$ or has a
		neighbor in $S$.} in $t$-pseudo-disk graphs. Even for $t$-intervals
	the best known previous approximation was $t^2$
	\cite{ButmanHLR10}. We observe that, via a simple reduction for Hypergraph Vertex
	Cover, it is $\mathsf{NP}$-hard to approximate dominating set to within a factor
	better than $t-1-\epsilon$ even in $t$-interval graphs. Under the Unique Games Conjecture, this lower bound can be slightly improved to $t-\epsilon$.
	\item An $\Omega(1/t)$ approximation for the maximum weight
	independent set in $t$-pseudo-disk graphs. A $1/(2t)$-approximation
	was known for $t$-intervals \cite{Bar-YehudaR06}. For more general
	shapes such as disks and fat objects, the approximation bounds in \cite{KammerT14}
	depend on $t$ and fatness parameters. Our approach generalizes to
	packing $t$-pseudo-disks into capacitated points.
\end{itemize}

We note that although our results are obtained via simple reductions to
existing algorithms in geometric packing and covering, these algorithms use fairly sophisticated ideas such as union complexity and Quasi-Uniform Sampling. Consequently, in some cases, the constants in the approximation guarantees may be worse compared to the known results. However, our results are applicable for wider class of geometric objects. Furthermore, it may be possible to improve the constants for the special case of intervals.

We describe the necessary geometric background in the next subsection. The three problems of
interest in this paper are minimum weighted set cover (MWSC), minimum weight
dominating set (MWDS) and maximum weight independent set (MWIS). We use
MSC, MDS and MIS to refer to the unweighted versions.

\subsection{Background from Geometric Approximation via LP Relaxations}
The approximability of MWSC, MWDS and MWIS in general graphs is
well-understood with essentially tight upper and lower bounds known.
In particular it is $\NP$-hard to obtain constant approximations for
MWDS and MWIS problems. However, in various geometric setting it is
possible to obtain improved algorithms including approximation schemes
(PTASes and QPTASes \cite{ChanH12,GibsonP10, ChanGKS12,AdamaszekW13,RoyGRR18,AdamaszekHW17}) and constant factor approximations
\cite{ChanH12,ChanGKS12,GibsonP10}. In this paper we are interested in LP-based approximations
for MWSC and MWIS that have been established via techniques that rely
on \emph{union complexity} of the underlying geometric objects. Union
complexity measures the worst-case representation size of the union of
a given collection of objects of a particular type or shape. In the
setting of planar obejcts the typical measure is the number of vertices
in the arrangement that appear on the boundary of the union. It is
well-known that many geometric objects such as intervals (on a line),
disks and squares (in the plane) have linear union complexity. In
fact, this holds for an even larger class of geometric objects, namely
pseudo-disks \cite{pach2008state}.

Bounds on union complexity have been used to obtain constant factor
and sub-logarithmic approximations for geometric MWSC and its
variants (see \cite{ClarksonV07,ChekuriCH12,Varadarajan10,ChanGKS12} and
\cite{MustafaV17} for a survey). Chan and Har-Peled \cite{ChanH12}
showed how union complexity can also be used to obtain improved
approximations for the MWIS problem. They give an LP rounding
algorithm with approximation guarantee $\Omega(n/u(n))$ for computing
an MWIS of $n$ objects with union complexity $u(\cdot)$.  We use this
result to give $\Omega(1/t)$-approximation for computing MWIS of
$t$-pseudo-disks. We note that although this implies
$\Omega(1/t)$-approximation for MWIS of $t$-intervals, $t$-disks and
$t$-squares; these results were already known
\cite{Bar-YehudaHNSS06,KammerTV10}. However, they use certain
``fatness'' properties of the underlying geometric objects, which do
not extend to pseudo-disks.

\medskip \noindent \emph{Shallow-cell complexity} (SCC) is a
generalization of the notion of union complexity to abstract set
systems \cite{ChanGKS12}. Low SCC has been used to obtain improved
approximations for MWSC and MWDS in geometric settings as well some
combinatorial settings.  Here, the general approach is to round a
feasible LP solution using a framework called \emph{Quasi-Uniform
	Sampling} introduced by Varadarajan \cite{Varadarajan10} originally
for geometric settings, and further refined and improved by Chan \etal
\cite{ChanGKS12}. We use this framework and the result in
\cite{ChanGKS12} to obtain an $O(t \log t)$-approximation for MWDS of
$t$-pseudo-disks by appropriately bounding the SCC of a related
instance via known results on MWDS of disks and pseudo-disks
\cite{GibsonP10,AronovDEP18}. This coupled with other ideas yields our
results on MWSC and MWDS for $t$-pseudo-disks.

\medskip
\noindent
\textbf{Orgnization.}  Section \ref{sec:prelims} introduces some
relevant notation and definitions. Section \ref{sec:covering}
describes algorithms for covering problems MWDS and MWSC
with a focus on the more involved MWDS problem. Section
\ref{sec:packing} describes algorithms for MWIS and a generalization.

\section{Preliminaries}
\label{sec:prelims}
A $t$-object is defined as the union of at most $t$ geometric
objects, where $t$ is a positive integer. Without loss of generality,
we assume that each $t$-object is a union of exactly $t$ objects. We
are typically interested in the case when the base objects come from a
specific class of geometric shapes, such as intervals, disks, and
pseudo-disks.

\begin{definition}[Pseudo-disks and $t$-pseudo-disks]
  A family $\mathcal{S}$ of connected regions bounded by simple closed
  curves in general position in the plane is called a collection of
  \emph{pseudo-disks}, if the boundaries of any distinct
  $S_i, S_j \in \mathcal{S}$ intersect at most twice.  Let
  $\mathcal{S}$ be a collection of pseudo-disks, and let
  $\mathcal{S'}$ be a collection of objects obtained by taking union
  of at most $t$ objects from $\mathcal{S}$. Then, we say that
  $\mathcal{S'}$ is a collection of $t$-pseudo-disks.
\end{definition}

Consider a $t$-object $O_i$. We use
$o_i^{(1)}, o_i^{(2)}, \ldots, o_i^{(t)}$ to denote the objects whose
union equals $O_i$. Then, by slightly abusing the notation, we also
denote by $O_i$ the set $\{o_i^{(1)}, o_i^{(2)}, \ldots,
o_i^{(t)}\}$. Note that, in some cases it may be desirable to require
that the objects $o_i^{(1)}, o_i^{(2)}, \ldots, o_i^{(t)}$ are
pairwise disjoint. However we do not impose such a restriction on the
objects and allow a more general model where they may intersect.

Next, we formally define the notion of union complexity, which is used to obtain the improved approximation for MWIS in Section 4.

\begin{definition}[Union complexity]
	Let $\S$ be a set of geometric objects in $\mathbb{R}^2$. If for any subset of objects $\R \subseteq \S$, the number of arcs on the boundary of the union of objects in $\R$ has at most $u(|\R|)$ vertices, then $\S$ is said to have union complexity $u(\cdot)$. We assume that $u(n) \ge n$ for any $n \ge 1$ and that $u$ is non-decreasing.
\end{definition}

\section{Minimum Weight Dominating Set and Set Cover}
\label{sec:covering}
As in Butman et al.\ \cite{ButmanHLR10}, we consider the following
slight generalization of the Minimum Dominating Set problem. Here, we
are given a undirected bipartite intersection graph of
$t$-objects. More formally, $G = (\R \sqcup \B, E)$, where
$\R = \{R_1, \ldots, R_N\}$ and $\B = \{B_1, \ldots, B_M\}$ are sets
of red and blue $t$-objects respectively. There is an edge between
$R_i \in \R$ and $B_j \in \B$ if $R_i \cap B_j \neq \emptyset$. Since
each vertex corresponds to a (red or blue) $t$-object, we will use the
term `vertex' and `$t$-object' interchangeably. Each $R_i \in \R$ has
a non-negative weight $w_i$. The goal is to find a subset
$\R' \subseteq \R$ of minimum weight, such that every blue vertex
$B_j$ has a neighbor in the set $\R'$. Note that the standard MWDS
problem is equivalent to the setting where $\R = \B$. In the
following, we consider the more general setting where $\R$ and $\B$
may be different.

\subsection{LP Relaxation and Rounding}
Let $\I = (\R, \B)$ denote the given instance of the (generalized)
MWDS. In an integer programming formulation we have a $\{0,1\}$
variable $x_i$ corresponding to a red $t$-object $R_i$ which is
intended to be assigned $1$ if $R_i$ is selected in the solution, and
is assigned $0$ otherwise. We relax the integrality constraints and
describe the natural LP relaxation below.
\begin{mdframed}[backgroundcolor=gray!9]
	\begin{alignat}{3}
	\text{minimize}\quad \displaystyle&\sum\limits_{R_i \in \R} w_i x_i & \nonumber \\
	\text{subject to}\quad \displaystyle&\sum\limits_{R_i: B_j \cap R_i \neq \emptyset} x_{i} \geq 1, \quad & \forall B_j \in \B \label[constr]{constr:coverage}\\
	\displaystyle & \qquad x_i \in [0, 1], \quad & \forall R_i \in \R\label[constr]{constr:fractional-x}
	\end{alignat}
\end{mdframed}

Let $x$ be an optimal LP solution for the given instance
$\I = (\R, \B)$. Since $x$ is feasible, for every blue $t$-object
$B_j$, we have
$\displaystyle \sum_{R_i: B_j \cap R_i \neq \emptyset} x_i \ge 1$. For
each $B_j \in \B$, let $b_j^{(k)} \in B_j$ denote an object maximizing
the quantity
$\displaystyle \sum_{R_i: b_j^{(k)} \cap R_i \neq \emptyset} x_i$,
where the ties are broken arbitrarily. For simplicity, we denote
$b_j^{(k)}$ by $b'_j$. Let $\B' \coloneqq \{b'_j: B_j \in \B\}$. Note
that for any $b'_j \in \B'$,
$\displaystyle \sum_{R_i : b'_j \cap R_i \neq \emptyset} x_i \ge 1/t$.

For any $R_i \in \R$, let $x'_i = \min\{tx_i, 1\}$, and let $x'$
denote the resulting solution. Then, for any $b'_j \in \B'$,
$\displaystyle \sum_{R_i: b'_j \cap R_i \neq \emptyset} x'_i \ge
1$. We emphasize that $\B'$ is a collection of ($1$-)objects, whereas
$\R$ is a collection of $t$-objects.  The following observation
follows from the definition of $x'$.

\begin{observation}\label{obs:soln}
  The solution $x'$ is feasible for the instance $\I' = (\R, \B')$, and
  we have that $\sum_{R_i \in \R} w_i x'_i \le t \cdot \sum_{R_i \in \R} w_i x_i.$
\end{observation}

The preceding step to reduce to $\I'$ is essentially the same as in \cite{ButmanHLR10}.

\paragraph*{Shallow-Cell Complexity}
Now, we describe how to round $x'$ to an integral solution via
Quasi-Uniform Sampling technique. To discuss this, we first consider
the closely related Minimum Weight Set Cover problem (MWSC).

Let $X$ be a set of $M$ elements, and $\mathcal{S}$ be a collection of
$N$ subsets of $X$. Each set $S_i \in \mathcal{S}$ has a non-negative
weight $w_i$. In MWSC the goal is to find a minimum-weight collection
of sets $\mathcal{S}'$ such that
$\bigcup_{S_i \in \mathcal{S}'} S_i = X$.

The standard LP relaxation for MWSC is as follows.
\begin{eqnarray*}
  \min \sum_{S_i \in \mathcal{S}} w_i x_i & & \\
  \sum_{S_i \ni j} x_i & \ge & 1  \quad \forall j \in X \\
  x_i &  \ge & 0 \quad \forall S_i \in \mathcal{S}
\end{eqnarray*}
Let $A \in \{0, 1\}^{M \times N}$ denote the constraint matrix in the
LP relaxation above. Each row of $A$ corresponds to an element, and
each column corresponds to a set. The entry $A_{ij} = 1$ if the
element $j$ is contained in $S_i$, otherwise $A_{ij} = 0$.

The following crucial definition is from \cite{ChanGKS12}.

\begin{definition}[Shallow-Cell Complexity]
  Let $1 \le k \le n \le N$. Let $S$ be any set of $n$ columns and let
  $A_S$ be the matrix restricted to the columns of $S$. If the number
  of distinct rows in $A_S$ with at most $k$ ones is bounded by
  $f(n, k)$ for any choice of $n, k$ and $S$, then the instance is
  said to have the shallow-cell complexity $f(n, k)$.
\end{definition}

Using bounds on the Shallow-Cell Complexity of an MWSC instance, it is
possible to round a feasible LP solution using a technique known as
Quasi-Uniform Sampling \cite{ChanGKS12,Varadarajan10}.

\begin{theorem}[\cite{ChanGKS12,Varadarajan10}] \label{thm:qus}
  Consider an MWSC instance with Shallow-Cell Complexity
  $f(n, k) = n\phi(n) \cdot k^c$, where $\phi(n) = O(n)$, and
  $c \ge 0$ is a constant. Then, there exists an algorithm to round
  any feasible LP solution for this instance within a factor of
  $O(\max\{\log \phi(N), 1 \})$, where the constant hidden in Big-Oh
  notation depends on the exponent $c$.
\end{theorem}

Usually, $\phi(n)$ is a function of $n$ such that $\phi(n) =
O(n)$. However, the same guarantee holds even when $\phi$ is
independent of $n$ . When we apply this theorem, we will set
$\phi(n) = O(t^4)$, which is independent of $n$.

Now, we consider an instance $\I' = (\R, \B')$ of the MWDS problem
obtained in the previous section. Recall that $\R$ is a collection of
$t$-objects and $\B$ is a collection of $1$-objects. Now, let
$\R' = \{r^{(k)}_i \in R_i: R_i \in \R \}$ denote the collection of
constituent red $1$-objects from $\R$. Let $\I'' = (\R', \B')$ denote
the MWDS instance thus obtained. Notice that an MWDS instance can also
be thought of as a MWSC instance. We first prove the following simple
lemma.

\begin{lemma} \label{lem:scc} Let $\I', \I''$ be MWDS instances as
  defined above. If the shallow-cell complexity of $\I''$ is
  $f(n, k)$, then the shallow-cell complexity of the corresponding
  instance $\I'$ is $g(n, k) \le f(nt, kt)$ for any
  $1 \le k \le n \le N$.
\end{lemma}
\begin{proof}
  We prove this fact from the definition of the shallow-cell
  complexity. Let $A^{\I'}$ denote the constraint matrix corresponding
  to the instance $\I'$. Fix some positive integers $n, k$ such that
  $1 \le k \le n \le |\R|$, and fix a set $S \subseteq \R$ of rows
  (i.e., $t$-objects), where $|S| = n$. Let $A^{\I'}_S$ denote the
  constraint matrix restricted to the columns corresponding to
  $S$. Let $P$ denote the set of distinct rows (i.e., blue objects) in
  $A^{\I'}_S$ with at most $k$ ones. We seek to bound $|P|$.

  Let $S' = \{r^{(k)}_i \in R_i : R_i \in S \}$ be the corresponding
  constituent $1$-objects. Note that $S' \subseteq \R'$, and
  $|S'| \le nt$. Let $A^{\I''}$ denote the constraint matrix
  corresponding to the instance $\I_1$ and let $A^{\I''}_{S'}$ denote
  $A^{\I''}$ restricted to the columns of $S'$. Let $P'$ denote the
  set of rows in $A^{\I''}_{S'}$ with at most $kt$ ones. Note that, if
  a row has at most $k$ ones in $A^{\I'}_S$, then it corresponds to
  exactly one row in $P'$. Therefore, $|P| \le |P'| \le f(nt, kt)$,
  where the last inequality follows from the definition of the
  Shallow-Cell Complexity of the instance $\I''$.
\end{proof}

Now, we state the known bounds on the Shallow-Cell Complexity of an
MWDS instance induced by red and blue collections of pseudo-disks.

\begin{theorem}[\cite{AronovDEP18}] \label{thm:pseudodisks} The
  Shallow-Cell Complexity of an MWDS instance induced by collections
  of pseudo-disks is at most $f(n, k) = O(nk^3)$.
\end{theorem}

Combining Theorems \ref{thm:qus} and \ref{thm:pseudodisks} and Lemma
\ref{lem:scc}, we obtain the following result.

\begin{theorem}
	There exists a randomized polynomial time $O(t \log t)$ approximation algorithm for the Red-Blue Dominating Set induced by a set of $t$-pseudo-disks.
\end{theorem}
\begin{proof}
  Let $\I = (\R, \B)$ be the original instance and let $x$ be an
  optimal LP solution. Define instance $\I' = (\R, \B')$ and the
  corresponding LP solution $x'$ as before. From Observation
  \ref{obs:soln}, we have that
  $\sum_{R_i \in \R} w_i x'_i \le t \cdot \sum_{R_i \in \R} w_i
  x_i$. By Lemma \ref{lem:scc} the shallow-cell complexity of $\I'$ is
  $g(n, k) \le f(nt, kt)$, where $f(n, k) = O(nk^3)$ by Theorem
  \ref{thm:pseudodisks}. Therefore, $g(n, k) \le O(nt^{4} k^3)$. Now,
  using the algorithm from Theorem \ref{thm:qus} with
  $\phi(n) = O(t^4)$, we can round $x'$ to an integral solution of
  cost at most
  $O(\log t) \cdot \sum_{R_i \in \R} w_i x'_i \le O(t \log t) \cdot
  \sum_{R_i \in \R} w_i x_i$, where the inequality follows from
  Observation \ref{obs:soln}.
\end{proof}

\begin{remark}
  Suppose we have an instance of $(\mathcal{R},\mathcal{B})$ of
  generalized MWDS where each red object is a $t_R$-pseudo-disk and
  each blue object is a $t_B$-pseudo-disk. Then the preceding analysis
  can be extended to obtain an $O(t_B \log t_R)$ approximation. We note that for the analogous version of intervals, a $(t_B \cdot t_R)$-approximation was obtained in \cite{ButmanHLR10}.
\end{remark}

\paragraph*{Geometric MWSC with $t$-objects.}
Consider a Geometric MWSC instance $(X, \mathcal{S})$, where
$X$ is a set of points, and $\mathcal{S}$ is a collection of $N$ sets
induced by $t$-objects. Recall that the goal is to find a
minimum-weight collection of $t$-objects $\mathcal{S}'$ that
\emph{covers} the set of points $X$. Furthermore, suppose that the
underlying geometric ($1$)-objects define a Set Cover instance with
Shallow-Cell Complexity $f(n, k) = n k^c$, for constant $c$. This
includes geometric objects such as intervals on a line, and
(pseudo-)disks in the plane. Then, using similar arguments, one can
bound the shallow-cell complexity of the MWSC instance by
$O(nk^c t^{c+1})$. Then Theorem \ref{thm:qus} implies an
$O(\log t)$-approximation.

If the underlying geometric objects are a set of fat triangles, then
the Shallow-Cell Complexity can be bounded by
$f(n, k) = n \log^* n \cdot k^c$, for some constant $c$
\cite{AronovBES14,ChanGKS12}. Using similar arguments, we can bound
the Shallow-Cell Complexity of the MWSC instance by
$f(nt, kt) = n \log^*(nt) \cdot t^{c+1} k^c$. Then, Theorem
\ref{thm:qus} implies an
$O(\log(t^{c+1} \cdot \log^*(Nt))) = O(\log t + \log \log^*
Nt)$-approximation. Therefore, we obtain the following result.

\begin{theorem}
  There exists an $O(\log t)$ approximation algorithm for covering a
  set of points by $t$-objects, where the underlying geometric objects
  are intervals on a line, or (pseudo-)disks in the plane. If the
  underlying objects are fat triangles in plane, we get a
  $O(\log t + \log \log^* Nt)$-approximation.
\end{theorem}

We note that this result improves on $t$-approximation for covering
points by $t$-intervals, which follows from the result of Hochbaum and
Levin \cite{HochbaumL06}. We also note that $O(\log t)$ is tight up to
constant factors even for covering points by $t$-intervals, which
follows from a reduction from the Set Cover problem as observed in
\cite{ButmanHLR10}.

\subsection{Integrality of MWDS LP for Intervals}
In this subsection, we prove that the MWDS LP for intervals is
integral. Let $\I = (\R, \B)$ be an MWDS instance, where $\R$ and $\B$
are collections of intervals on a line. Butman et al.\
\cite{ButmanHLR10} proved this using a primal-dual algorithm that
constructs an integral solution of the same cost as the LP. We give a
simpler proof of this fact by using structure of the constraint
matrix.

First, we preprocess blue intervals such that no blue interval is
completely contained inside another blue interval, and let $\I'$
denote the resulting instance. Suppose $B_1, B_2 \in \B$ are two
intervals such that $B_1 \subset B_2$. Then, any feasible integral
solution must include a red interval $R$ that intersects $B_1$, and
thus $B_2$. Furthermore, the LP constraint corresponding to $B_2$ is
implied by the constraint corresponding to $B_1$. Therefore, the
feasible regions of the LP's corresponding to $\I$ and $\I'$ are also
the same.

Now, in the following theorem we show that the constraint matrix of
the LP corresponding to $\I'$ satisfies the consecutive ones property,
and thus it is totally unimodular. This implies that the LP
corresponding to $\I_1$ (and therefore $\I$) is integral (see, e.g.,
\cite{schrijver2003combinatorial}).

\begin{theorem}
  Let $\I = (\R, \B)$ be an MWDS instance, where $\R$ and $\B$ are
  collections of ($1$)-intervals. Then, the MWDS LP is integral.
\end{theorem}
\begin{proof}
  Let $A$ denote the constraint matrix corresponding to the
  preprocessed instance, where the rows (i.e., blue intervals) are
  sorted in the non-decreasing order of their left endpoints. We show
  that $A$ has the consecutive ones property for any column.

  Consider any column (i.e., a red interval) $R$. Let $B_j$ and $B_k$
  denote the first (leftmost) and last (rightmost) blue intervals
  intersecting $R$ in this order respectively. Note that $j \le
  k$. Now, suppose for contradiction that there is an interval
  $B_\ell$ with $j < \ell < k$ that does not intersect $R$. Therefore,
  we have that
  $\textsf{Left}(B_j) < \textsf{Left}(B_\ell) \le
  \textsf{Right}(B_\ell) < \textsf{Left}(R) \le
  \textsf{Right}(B_j)$. Here, the third and fourth inequalities follow
  from the assumptions that $R \cap B_\ell = \emptyset$ and
  $R \cap B_j \neq \emptyset$ respectively. However, this implies that
  $B_\ell \subset B_j$, which is a contradiction. Thus, the column
  corresponding to $R$ satisfies the consecutive ones property.
\end{proof}

\subsection{Hardness Results for MWDS}
We show that it is $\mathsf{NP}$-hard to obtain a $t-1-\epsilon$
approximation for MWDS problem for $t$-intervals, for any $t \ge 3$
and $\epsilon > 0$. The lower bound can be improved to $t-\epsilon$
for any $t \ge 2$ and $\epsilon > 0$, assuming the Unique Games
Conjecture (UGC). In fact, this result follows even in the very
restricted setting, where $\R$ is a collection of points, and $\B$ is
a collection of $t$-points, i.e., a union of at most $t$ distinct
points. Note that this is a special case where $\R$ and $\B$ are
collections of $t$-intervals.

Let $(X, \S)$ be an $f$-uniform instance of MWSC. That is, any element
in $X$ is contained in exactly $f$ sets of $\S$. We reduce this to the
MWDS problem as follows. For every set $S_i \in \S$, we add a distinct
red point $R_i$ on a line, and set its weight equal to that of
$S_i$. For any element $e_j \in X$, we add at most $f$ blue points
coinciding with points $R_i$, where $e_j \in S_i$. Note that a blue
$t$-point is ``covered'' if we select a red point that coincides with
its constituent $t$ points. Thus, there is a one-to-one correspondence
between a feasible solution to the original set cover instance, and a
feasible solution to the MWDS instance. We have the following hardness
results for the $f$-uniform MWSC problem.

\begin{theorem}
  It is $\mathsf{NP}$-hard to obtain a $f-1-\epsilon$ approximation
  for $f$-uniform Set Cover, for any $f \ge 3$ and $\epsilon > 0$
  \cite{DinurGKR05}.  \\Assuming UGC, the lower
  bound can be improved to $f-\epsilon$ for any $f \ge 2$ and
  $\epsilon > 0$ \cite{KhotR08}.
\end{theorem}

Therefore, from the above reduction, we get the following hardness
results for the MWDS problem.

\begin{theorem}
  It is $\mathsf{NP}$-hard to obtain a $t-1-\epsilon$ approximation
  for the special case of MWDS, where $\R$ is a set of points and $\B$
  is a set of $t$-points.  \\Assuming UGC, the lower bound can be
  improved to $t-\epsilon$ for any $t \ge 2$ and $\epsilon > 0$.
\end{theorem}

\section{Maximum Weight Independent Set}
\label{sec:packing}
We consider MWIS of $t$-objects. Here, we are given a set $\S = \{S_1, \ldots, S_n\}$\footnote{In this section, we assume that we are given the geometric representation of the $t$-objects in $\S$.}, where each $S_i \in \S$ is a $t$-object, and has a non-negative weight $w_i$. The goal is to find maximum weight independent set $\S' \subseteq \S$. That is, for any distinct $S_i, S_j \in \S'$, we must have $S_i \cap S_j = \emptyset$.

Let $\mathcal{A}(\S)$ denote the arrangement of $\S$, and let
$\mathcal{V}(\S)$ denote the set of vertices in $\mathcal{A}(\S)$. We
use the LP rounding algorithm for Maximum-Weight Independent Set from
\cite{ChanH12}. First, we state the LP relaxation from \cite{ChanH12}.

\begin{mdframed}[backgroundcolor=gray!9]
	\begin{alignat}{3}
	\text{maximize}\quad \displaystyle&\sum\limits_{S_i \in \S} w_i x_i & \nonumber \\
	\text{subject to}\quad \displaystyle&\sum\limits_{S_i \ni p} x_{i} \leq 1, \quad & \forall p \in \mathcal{V}(\S) \label[constr]{constr:indep}\\
	\displaystyle & \qquad x_i \in [0, 1], \quad & \forall S_i \in \S\label[constr]{constr:indep-fractional-x}
	\end{alignat}
\end{mdframed}

We have a simple observation that relates the union complexities of
$t$-objects and the corresponding ($1$-)objects.

\begin{observation} \label{obs:union} Let $\S$ be a collection of
  $t$-objects, and suppose the collection of underlying ($1$-)objects
  has union complexity $u(\cdot)$. Then, the union complexity of any
  $k$ objects in $\S$ is at most $u(kt)$.
\end{observation}
\begin{proof}
  Let $\R \subseteq \S$ be any subset of $t$-objects. Let
  $\R' = \{ r^{(k)}_i \in R_i : R_i \in R \}$ denote the underlying
  set of ($1$)-objects. Note that $|R'| \le t \cdot |R|$. Since the
  underlying collection of ($1$)-objects has union complexity
  $u(\cdot)$, the number of arcs on the boundary of $R'$ is at most
  $u(|R'|) \le u(t \cdot |R|)$.
\end{proof}

Chan and Har-Peled \cite{ChanH12} give a randomized LP rounding algorithm with the following guarantee.

\begin{theorem} \label{thm:mis} There is an
  $\Omega(n/u(n))$-approximation algorithm for an MWIS instance, for a
  collection of $n$ geometric objects with union complexity
  $u(\cdot)$.
\end{theorem}

Combining Observation \ref{obs:union} and Theorem \ref{thm:mis}, we get the following result.
\begin{theorem}
  There is an $\Omega(n/u(nt))$-approximation algorithm for an
  MWIS instance for an instance of $t$-objects, where the underlying
  ($1$-)objects have union complexity $u(\cdot)$.
\end{theorem}

In particular, since pseudo-disks have linear union complexity
\cite{pach2008state}, an $\Omega(1/t)$-approximation follows for MWIS
of $t$-pseudo-disks. As noted in the introduction, this result is
tight up to constant factors even for the case of $t$-intervals.

\paragraph*{Packing $t$-objects}
We consider a related problem called Maximum Weight Region Packing. We
are given a collection of $t$-objects $\mathcal{S}$, and a set of
points $P$. Each $t$-object $S_i$ has a weight $w_i$ and each point
$p \in P$ has a capacity $c(p)$, which is a positive integer. The goal
is to find a maximum-weight collection $\mathcal{S}'$ of $t$-objects,
such that for each point $p \in P$, the number of regions of
$\mathcal{S}'$ that contain $p$ is at most $c(p)$. Note that MWIS is a
special case when $P =\mathcal{V}(S)$ is the set of all points in
arrangement and $c(p) = 1$ for each $p \in P$.

Extending the LP rounding algorithm of Chan and Har-Peled
\cite{ChanH12} for the MWIS problem, Ene et al.\ \cite{EneHR17} gave
an $\Omega((\frac{n}{u(n)})^{1/C})$-approximation for Maximum Weight Region
Packing problem, where $C$ is the minimum capacity of any point. Using
similar arguments as in MWIS, we obtain the following result.

\begin{theorem}
  There exists a polynomial $\Omega(1/t^{1/C})$-approximation
  algorithm for Maximum Weight Region Packing with $t$-pseudo-disks.
\end{theorem}

\section{Concluding Remarks}
In geometric settings, Quasi-Uniform Sampling has been used for Set
Multicover \cite{BansalP16}, and other related ideas have been used
for Partial Set Cover \cite{InamdarV18,ChekuriQZ19}. Our results for
MWDS and MWSC using $t$-objects can be extended to these settings, and
we leave these extensions for a future version of the paper. We also
note that even for the special case of MWDS of $t$-intervals, there is
a gap between $O(t \log t)$-approximation and a lower bound of
$\Omega(t)$, as shown in Section \ref{sec:covering}. Resolving this
gap is an interesting open question.

\bibliography{multiple.bib}

\end{document}